\documentclass[letterpaper, 10 pt, journal, twoside]{IEEEtran}

\IEEEoverridecommandlockouts                              

\usepackage{amsthm}
\usepackage{amsmath}
\usepackage{amsfonts}
\usepackage{mathtools}
\usepackage{color}
\usepackage{algorithm}
\usepackage{algorithmic}
\usepackage{cite}
\usepackage{hyperref}

\newtheorem{definition}{Definition}
\newtheorem{proposition}{Proposition}

\newtheorem{remark}{Remark}

\usepackage{layouts}
\usepackage{booktabs}
\usepackage{threeparttable}

\newcommand{\R}{\mathbb{R}}
\newcommand{\nstates}{{n_x}}
\newcommand{\ninputs}{{n_u}}

\begin{document}
%
\title{
Interpreting and Improving Optimal Control Problems with Directional Corrections
}

\author{Trevor Barron and Xiaojing Zhang
\thanks{© 2025 IEEE.  Personal use of this material is permitted.  Permission from IEEE must be obtained for all other uses, in any current or future media, including reprinting/republishing this material for advertising or promotional purposes, creating new collective works, for resale or redistribution to servers or lists, or reuse of any copyrighted component of this work in other works.}
\thanks{Manuscript received: November 20, 2024; Revised February 18, 2025; Accepted March 20, 2025.
This paper was recommended for publication by Editor Lucia Pallottino upon evaluation of the Associate Editor and Reviewers' comments.
{\it  (Trevor Barron and Xiaojing Zhang contributed equally to this work.) (Corresponding authors: T. Barron; X. Zhang.)}}
\thanks{The authors are with Apple Inc., Cupertino, CA 95014, USA.
{\tt\footnotesize\ \{trevor\_barron, xiaojing\_zhang2\}@apple.com}.}
\thanks{Digital Object Identifier (DOI): 10.1109/LRA.2025.3557226}
}

\markboth{IEEE Robotics and Automation Letters. Preprint Version. Accepted March, 2025}
{Barron \MakeLowercase{\textit{et al.}}: Interpreting and Improving Optimal Control Problems with Directional Corrections} 

\maketitle

\begin{abstract}
Many robotics tasks, such as path planning or trajectory optimization, are formulated as optimal control problems (OCPs). 
The key to obtaining high performance lies in the design of the OCP's objective function.
In practice, the objective function consists of a set of individual components that must be carefully modeled and traded off such that the OCP has the desired solution.
It is often challenging to balance multiple components to achieve the desired solution and to understand, when the solution is undesired, the impact of individual cost components.
In this paper, we present a framework addressing these challenges based on the concept of \emph{directional corrections}. 
Specifically, given the solution to an OCP that is deemed undesirable, and access to an expert providing the direction of change that would increase the desirability of the solution, our method analyzes the individual cost components for their ``consistency" with the provided directional correction.
This information can be used to improve the OCP formulation, e.g., by increasing the weight of consistent cost components, or reducing the weight of -- or even redesigning -- inconsistent cost components.
We also show that our framework can automatically tune parameters of the OCP to achieve consistency with a set of corrections.
\end{abstract}

\begin{IEEEkeywords}
Optimization and Optimal Control,
Motion and Path Planning, 
Machine Learning for Robot Control, 
Autonomous Agents
\end{IEEEkeywords}

\section{Introduction}

%

\IEEEPARstart{A}{n} ever-increasing number of tasks in robotics and controls are formulated as optimal control problems (OCPs), where a robot's desired behavior is characterized as the solution of an OCP. Examples include motion and path planning \cite{lavalle2006planning, murray2009optimization, padenFrazzoli2016survey, Zhang2021_OBCA}, building climate control \cite{oldewurtel2012buildingClimateControl, ZhangECC2013_BuildingControl, maddalena2020data} or  battery management \cite{moura2010stochastic, chaturvedi2010algorithms, guanetti2017optimal}.
Robotic tasks are often executed in a model predictive control (MPC) or receding-horizon fashion, where a finite-horizon OCP is solved at each timestep. 
Upon execution of the first control input, the system transitions to a new state at which a new OCP is solved \cite{Borrelli_Bemporad_Morari_MPC_book_2017, Mayne_nonlinearMPC2000}.
This cycle repeats until the robot has completed its task, e.g., arrived at its goal.

A typical OCP can be divided into two parts: a system model with its operational domain, which enters the OCP as constraints, and an objective function that the robot optimizes for while satisfying the constraints.
When the system is not operating near its limits, a robot's behavior is determined by the objective function and the system model. 
In practice, the objective function is often formulated as a weighted sum of individual components that are traded off against each other
\cite{lavalle2006planning, murray2009optimization, padenFrazzoli2016survey}.
While an OCP's objective function and system model concisely define the robot's optimal behavior, interpreting the contribution of individual components to the optimal behavior is often challenging in practice, as the objective function may contain dozens, possibly even hundreds, of components.
The interplay between many cost components and the system model makes analysis of -- and improvements to -- the OCP a painstaking task that involves significant trial and error. 

Consider, for example, the case where a robot's plan is defined as the (locally) optimal solution to an OCP.
Suppose an oracle (e.g., a human expert) observes that the plan is undesirable, and provides a direction of desired change in the solution \cite{jin2022learning}.
An example \emph{directional correction} could be ``the robot should move slower", or ``the robot should stay further to the left". 
To understand the cause of the undesired behavior, and to translate those directional corrections into improvements in the OCP, practitioners need insights into \emph{(i)} Which cost components are consistent or inconsistent with a directional correction? \emph{(ii)} How can one adjust the weights on those cost components to improve the plan in the direction defined by the correction? 
Surprisingly, these questions turn out to be difficult to answer. 

While it is common in practice to associate the impact of each cost component with their magnitudes (e.g., cost components that show high values are often singled out as root causes for undesired behavior), this approach often leads to misleading conclusions, as we will demonstrate in Section~\ref{sec:experiment} of this paper. 
Furthermore, due to system dynamics, it is possible that an undesired behavior at a certain stage of the plan (e.g., a slow-down at 1 second) is caused by a cost at an entirely different stage (e.g., one associated with lateral movement five seconds ahead). 
These cross-stage effects can be counter-intuitive and challenging to understand.

In this paper, we present an OCP analysis framework that systematically addresses the following question: Given a plan, which cost components and constraints are in conflict (inconsistent) or agreement (consistent) with the desired change defined by a \emph{directional correction}?
Based on this framework, we also demonstrate a method that finds parameters of the OCP to maximize consistency with the corrections.
Our contributions are as follows:
\begin{itemize}
    \item We introduce the notion of \emph{consistency} and present a framework that analyzes the consistency of each cost component with respect to a directional correction. 
    We extend the work of \cite{jin2022learning}, which {focuses on} corrections with respect to control inputs, to allow corrections formulated in terms of states. This significantly increases the applicability of the method, since corrections are often easier to provide on states such as position, speed or acceleration, rather than control inputs such as jerk. 
    Our framework naturally considers the cross-stage effects of cost functions, and we demonstrate that it provides a more accurate representation of the influence of a cost component than the traditional cost magnitude analysis.
    We also present an extension to reason about the consistency of inequality constraints with respect to a directional correction.

    \item We show that our framework applies to both open-loop and closed-loop settings, where the latter is relevant in a receding-horizon MPC setup.

    \item We demonstrate how to use the framework to automatically tune cost component weights in an OCP, assuming access to an oracle that provides a dataset of corrections.


\end{itemize}

{The impact of cost components on the solution of an OCP has also been studied through the lens of sensitivity analysis \cite{boyd2004convex, gal2012advances} and parametric optimization \cite{bemporad2002explicit, borrelli2017predictive} by characterizing the optimal plan as a function of the weights in the cost function. 
Parametric optimization has mostly focused on linear systems with convex costs due to the complexity of computing, let alone explicitly representing, the optimal plan for non-linear systems.
In contrast, our approach is straightforward to apply to non-linear systems with non-convex cost functions since we do not require an explicit representation of the optimal plan with respect to parameters.}

\section{Problem Setup}
We assume that the robot's dynamics are described by 
\begin{equation}\label{eq:dynamics}
    x_{k+1} = f(x_k, u_k),
\end{equation}
where $x_k \in \R^{\nstates}$ is the robot's state at timestep $k$ given an initial state $x_0 = x_\text{init}$, $u_k \in \R^{\ninputs}$  is the control input, and $f: \R^{\nstates} \times \R^{\ninputs} \to \R^{\nstates}$ describes the robot's dynamics. The system is subject to input and state constraints of the form $h(x_k, u_k) \leq 0$, where $h: \R^{\nstates} \times \R^{\ninputs} \to \R^{n_h}$, $n_h$ is the number of constraints, and the inequality constraint is interpreted elementwise. The planning problem amounts to finding a sequence of states $x_{0:N}:=(x_0,x_1,\ldots, x_N)$ and inputs $ u_{0:N-1}:= (u_0,u_1,\ldots,u_{N-1})$, over a horizon $N$, that optimizes a given objective function
\begin{subequations}
    \begin{equation}\label{eq:cost}
        J(x_{0:N}, u_{0:N-1}) = \ell_N(x_N) + \sum_{k=0}^{N-1} \ell_k(x_k,u_k),
    \end{equation}
where $\ell_k:\R^\nstates \times \R^\ninputs \to \mathbb \R$ are differentiable stage costs, and $\ell_N:\R^\nstates \to \R$ is the terminal cost. We assume that the stage cost $\ell_k(\cdot,\cdot)$ is the weighted sum of $R$ individual cost components $\ell_k^{(r)}:\R^\nstates \times \R^\ninputs \to \mathbb \R$, such that
    \begin{equation}\label{eq:cost_components}
        \ell_k(x_k, u_k) = \sum_{r = 0}^R w_k^{(r)} \ell_k^{(r)}(x_k, u_k),\\
    \end{equation}
\end{subequations}
where $w_k^{(r)}>0$ is the weight. 
Combining \eqref{eq:dynamics} -- \eqref{eq:cost_components}, the OCP can be formulated as
\begin{align}\label{eq:OCP}
\min_{x_{0:N}, u_{0:N-1}} \quad & J(x_{0:N}, u_{0:N-1}) \nonumber \\
    \text{s.t.} \quad   &x_0 = x_{\text{init}} \\
                        &x_{k+1} = f(x_k, u_k), \quad k=0,1,\ldots,N - 1, \nonumber 
\end{align}
where, for ease of exposition, we omit the constraints $h(x_k,u_k) \leq 0$. Throughout the paper, we refer to $\zeta^\star = (x^\star_{0:N}, u^\star_{0:N-1})$ as the optimal plan. Furthermore, any $\zeta = (x_{0:N}, u_{0:N-1})$ satisfying the dynamics \eqref{eq:dynamics} is a plan.


\section{Consistency Analysis Framework}\label{sec:consistency_analysis_framework}
We begin by recalling the notion of \emph{directional corrections} from \cite{jin2022learning} in Section~\ref{sec:directional_corrections}. In Section~\ref{sec:consistency}, we formally define the notion of \emph{consistency}, for which we provide a technical analysis in Section~\ref{sec:consistency_analysis}. Extensions are discussed in Sections~\ref{sec:consistency_extension_closed_loop} and \ref{sec:consistency_extension_constraints}.

\subsection{Directional Corrections}\label{sec:directional_corrections}
Given a plan $\zeta = (x_{0:N}, u_{0:N-1})$,  we define the vector $a := (a_x,a_u) \in \R^{(N + 1)\nstates + N\ninputs}$ to be a \emph{directional correction} representing the direction in which we desire the plan to move. 
For example, if our system is a one-dimensional triple integrator with state $x = (\text{position}, \text{velocity}, \text{acceleration})$, input $u = \text{jerk}$, and horizon $N=1$, then the directional correction $((0,0,0),(0,0,1), (0))$ for the plan $\zeta=(x_0,x_1,u_0)$, implies that the robot's plan should have had higher acceleration at stage~1. 
The directional correction indicates only that the plan should have higher acceleration, but it does not quantify by how much.

Directional corrections may be obtained in different ways. For example, one might introspect a plan, detect undesirable behavior, and annotate a desired change.
Alternatively, a directional correction could be derived automatically from a set of system requirements, desired outcomes, or expert demonstrations. We refer the interested reader to \cite{jin2022learning} for detailed discussion on directional corrections.

\subsection{Consistency}\label{sec:consistency}
Let $\zeta = ( x_{0:N}, u_{0:N-1})$ be a plan satisfying the system dynamics \eqref{eq:dynamics}, and let $A_k := \frac{\partial f}{\partial x_k}(x_k,u_k)$ and $B_k := \frac{\partial f}{\partial u_k}(x_k,u_k)$ be the linearization of \eqref{eq:dynamics} about $\zeta$.
Let $F_{x,u} \in \R^{(N+1)\nstates \times N\ninputs}$ be the sensitivity matrix\footnote{$F_{x,u}$ can also be understood as the total derivative of $x_{0:N}$ with respect to $u_{0:N-1}$, under the system dynamics \eqref{eq:dynamics}.}
 \begin{equation*}
     F_{x,u} = 
     \begin{bmatrix} 
    0 & 0 & \cdots & 0  \\
    B_0 & 0 & \cdots & 0\\
    A_1 B_0 & B_1 & \cdots & 0\\
    A_2 A_1 B_0 & A_2 B_1 & \cdots & \vdots \\
    \vdots & \vdots & \ddots & 0\\
    A_{N-1} \cdots A_1 B_0 & A_{N-1} \cdots A_2 B_1 & \hdots & B_{N-1}
     \end{bmatrix}
 \end{equation*}
that maps a (small) change in control inputs $\Delta u_{0:N-1}$ to the change in the states $\Delta x_{0:N}$, i.e.,
\begin{equation}\label{eq:sensitivity}
    \Delta x_{0:N} \approx F_{x,u} \Delta u_{0:N-1}.
\end{equation}

We are now interested in deriving the effect of a given cost component $\ell_k^{(r)}(x_k, u_k)$ around $\zeta$.
To streamline the upcoming discussion, we re-express the cost components $\ell_k^{(r)}(x_k, u_k)$ in their eliminated forms $\tilde\ell_k^{(r)}(u_{0:N-1})$, obtained by eliminating the state variable $x_k$ using the system dynamics \eqref{eq:dynamics}.
This variable elimination allows us to reason explicitly about the effect of each control input on the cost $\ell_k^{(r)}(\cdot)$, without needing to consider the dynamics~\eqref{eq:dynamics}.

Let $\nabla \tilde\ell_k^{(r)}(\cdot) \in \R^{N \ninputs}$ denote the gradient of $\tilde\ell_k^{(r)}(\cdot)$. 
Define $F := \begin{bmatrix}F_{x,u}; I_{N\ninputs}\end{bmatrix} \in \R^{((N+1)\nstates+N\ninputs)\times N\ninputs}$ to be the matrix obtained by (vertically) stacking $F_{x,u}$ with the identity matrix $I_{N\ninputs} \in \R^{N\ninputs \times N\ninputs}$. 
We are now in place to introduce the concept of \emph{consistency score}.

\begin{definition}[Consistency Score]\label{def:consistency_score}
    Let $\zeta = (x_{0:N}, u_{0:N-1})$ be a plan, $a = (a_x, a_u)$ be a directional correction, and $F$ be as defined above. Then, the \emph{consistency score} of the eliminated cost $\tilde\ell_k^{(r)}(u_{0:N-1})$ is defined as $\text{cs}_k^{(r)} := \left<a,    -w_{k}^{(r)} F \nabla\tilde\ell_{k}^{(r)}(u_{0:N-1})\right>$.
\end{definition}
This allows us to define the notion of \emph{consistency}.
\begin{definition}[Consistency]\label{def:consistency}
    An eliminated cost $\tilde\ell_k^{(r)}(u_{0:N-1})$ is called \emph{consistent} if $\text{cs}_k^{(r)} > 0$, and \emph{inconsistent} if $\text{cs}_k^{(r)} < 0$.
\end{definition}

In general, consistency is a local property defined about a trajectory $\zeta = (x_{0:N}, u_{0:N-1})$. 
Roughly speaking, a cost component is consistent if, on its own and in the absence of any other costs, it would locally move the trajectory $\zeta = (x_{0:N}, u_{0:N-1})$ in the direction of $a = (a_x, a_u)$. To see it, we observe that $-\nabla \tilde{\ell}_k^{(r)}(u_{0:N-1})$ is the direction of steepest descent of $\tilde{\ell}_k^{(r)}(\cdot)$ at $u_{0:N-1}$. 
By virtue of \eqref{eq:sensitivity}, its effect on the state trajectory is $-F_{x,u}\nabla \tilde{\ell}_k^{(r)}(u_{0:N-1})$, while its effect on the plan is given by $-F\nabla \tilde{\ell}_k^{(r)}(u_{0:N-1})$. 
It now follows immediately from Definition~\ref{def:consistency} that, for a consistent cost, $-F\nabla \tilde{\ell}_k^{(r)}(u_{0:N-1})$ is in the same direction of $a$, thus pushing the plan in the direction of $a$.


\begin{remark}
    We point out that the stage cost, $\ell_k^{(r)}(x_k, u_k)$, while defined on the stage variables $(x_k, u_k)$, affects
    both earlier and later stages through the dynamics \eqref{eq:dynamics}, as $x_k$ depends on $u_{0:k-1}$ and $x_{k+1:N}$ depends on $u_k$.
    Therefore, minimizing $\ell_k^{(r)}(x_k, u_k)$ affects the entire plan $\zeta=(x_{0:N}, u_{0:N-1})$.
    This cross-stage coupling means an undesirable behavior at stage $k$ can be caused by a cost component at a different stage\footnote{This argument can be alternatively understood as a forward-backward pass deployed in DDP/iLQR-type of algorithms.}.
\end{remark}

\subsection{Properties of (In)consistent Costs}\label{sec:consistency_analysis}
In this section, we will formally show for linear systems that a consistent cost component enjoys the intuitive property that, if the corresponding weight is increased, then the plan will improve by moving in the direction defined by the correction.
Specifically, let $\zeta^\star = (x_{0:N}^\star, u_{0:N-1}^\star)$ be the optimal plan of the OCP~\eqref{eq:OCP}, $a=(a_x,a_u)$ be a desired directional correction, $\tilde\ell_k^{(r)}(u_{0:N-1}^\star)$ be a consistent cost component, and $w_k^{(r)}>0$ its associated weight. Consider now the modified cost function, obtained by increasing the weight by $\delta>0$: 
\begin{equation}\label{eq:new_cost_function}
    \hat J(u_{0:N-1}) :=  \tilde J(u_{0:N-1}) + \delta \tilde\ell_k^{(r)}(u_{0:N-1}).
\end{equation}

\begin{proposition}\label{prop:consistency_localImprovement}
Assume our system is linear, i.e., $x_{k+1} = A_k x_k + B_k u_k$. Then there exists a plan $\hat \zeta = (\hat x_{0:N}, \hat u_{0:N-1})$ that improves upon $\zeta^\star = (x_{0:N}^\star, u_{0:N-1}^\star) $ in terms of $\hat J(\cdot)$, that moves in the direction of the correction $a$ with respect to $\zeta^\star$.
\end{proposition}
\begin{proof}
    Consider the following candidate solution $\hat u_{0:N-1} := u_{0:N-1}^\star - \epsilon\delta\nabla\tilde\ell_k^{(r)}(u_{0:N-1}^\star)$. We will now demonstrate that $ - \epsilon\delta\nabla\tilde\ell_k^{(r)}(u_{0:N-1}^\star)$ is a descent direction by showing that $\delta\nabla\tilde\ell_k^{(r)}(u_{0:N-1}^\star)$ is the gradient of $\hat J(u_{0:N-1}^\star)$, and hence $\hat J(\hat u_{0:N-1}) < \hat J(u_{0:N-1}^\star)$, for some $\epsilon>0$. 
    It follows from the first-order optimality condition that $\nabla \tilde J(u_{0:N-1}^\star)=0$, such that $\nabla \hat J(u_{0:N-1}^\star) = \nabla \tilde J(u_{0:N-1}^\star) + \delta\nabla\tilde \ell_k^{(r)}(u_{0:N-1}^\star) = \delta \nabla \tilde \ell_k^{(r)}(u_{0:N-1}^\star)$.
    Furthermore, it follows from Definition~\ref{def:consistency} that $\nabla\tilde\ell_k^{(r)} \neq 0$, and hence $\hat J(\hat u_{0:N-1}) < \hat J(u_{0:N-1}^\star)$, for a sufficiently small $\epsilon>0$. To see that the new control trajectory $\hat u_{0:N-1}$ indeed moves the plan $\hat\zeta$ in the direction of the correction with respect to $\zeta^\star$, we observe from \eqref{eq:sensitivity}, which now holds with equality due to linearity assumption, that the change in control trajectory $\Delta u_{0:N-1} = -\epsilon\delta\nabla \tilde\ell_k^{(r)}(u_{0:N-1}^\star)$  changes the state trajectory as $\Delta x_{0:N} = -\epsilon\delta F_{x,u} \nabla \tilde\ell_k^{(r)}(u_{0:N-1}^\star)$, and thus the plan by $\Delta\zeta = (\Delta x_{0:N}, \Delta u_{0:N-1})$. Therefore, $\langle  a, \hat{\zeta} - \zeta^\star \rangle = \langle a, \Delta\zeta \rangle = \langle a, -\epsilon\delta {F}\, \nabla \tilde{\ell}_k^{(r)}(u_{0:N-1})\rangle > 0$, where the inequality follows from Definition~\ref{def:consistency}.
\end{proof}
Notice that Proposition~\ref{prop:consistency_localImprovement} does not establish that the optimal plan under new cost function \eqref{eq:new_cost_function}, $\hat{\zeta}^\star$, improves upon $\zeta^\star$ in the direction of $a$, i.e., it does not establish that $\langle a, \hat{\zeta}^\star - \zeta^\star   \rangle >0$. However, we have empirically observed that this is the case in practice. Furthermore, while the proof is given for linear systems, the same properties have been observed to hold in practice for nonlinear systems. We close the section by remarking that the reverse property can be shown for inconsistent cost components. 

\subsection{Extension to Closed-Loop Analysis}\label{sec:consistency_extension_closed_loop}

In practice, many systems employ an MPC scheme where they iteratively replan to respond to updates in the environment \cite{Borrelli_Bemporad_Morari_MPC_book_2017, Mayne_nonlinearMPC2000}. 
In this scheme, a system computes a motion plan at each step, but only executes the initial control input at each execution cycle.
We describe an extension to the previously discussed ``open-loop" case that permits analysis of an OCP using directional corrections expressed with respect to the closed-loop trajectory. 

For this section, we denote by $x_{k|t}$, $k\geq0$, the state at time $t+k$, predicted at time $t$, and by $u_{k|t}$ the control input at time $t+k$ predicted at time $t$. 
The true state at time $t$ is $x_t$.
In MPC, at time $t$, the controller only applies the first control input $u_{0|t}$ of the entire input sequence $u_{0:N-1|t} = (u_{0|t},\ldots,u_{N-1|t})$, while the remainder of the input sequence $u_{1:N-1|t}$ is discarded. 
At the next timestep, a new input sequence is computed $u_{0:N-1|t+1}$ based on new state measurement $x_{t+1}$, and again only the first control input $u_{0|t+1}$ is applied.
Accordingly, let $x_{0|0:T} := (x_{0|0}, \ldots, x_{0|T})$ denote a robot's closed-loop trajectory, obtained by applying the closed-loop input trajectory $u_{0|0:T-1} := (u_{0|0}, \ldots, u_{0|T-1})$. 
It is easy to verify that the closed-loop sensitivity matrix $F_{x,u}^\text{cl}$ takes the same form as $F_{x,u}$ described in Section~\ref{sec:consistency}, with $T$ (the duration of the closed-loop execution) in place of $N$ (the prediction horizon used when computing the open-loop plans). 

To evaluate the effect of a cost component $\tilde\ell_{k|t}^{(r)}(u_{0:N-1|t})$ on the closed-loop trajectory, let us consider its gradient with respect to only the first input $u_{0|t}$, which we denote by $\nabla_{u_{0|t}}\tilde\ell_{k|t}^{(r)}(u_{0:N-1|t})$.
Assuming we are interested in evaluating the overall effect of the $r$th cost component on the $k$th predicted stage, then we can collect the gradients from each motion plan as
$\Delta u_{0:T-1}^{r} := \begin{bmatrix*}\nabla_{u_{0|0}} \tilde{\ell}_{k|0}^{(r)}(u_{0:N-1|0}); \ldots; \nabla_{u_{0|T-1}} \tilde{\ell}_{k|T-1}^{(r)}(u_{0:N-1|T-1})\end{bmatrix*} \in \R^{T\ninputs}$.
Then we can estimate that cost component's local effect on the closed-loop state trajectory as $\Delta x_{0:T} \approx F_{x,u}^\text{cl} \Delta u_{0:T-1}$. 
Similar as in Section~\ref{sec:consistency}, we may now define $F^\text{cl} := [F_{x,u}^\text{cl} ; I_{T\ninputs}]$, and define \emph{closed-loop consistency score} 
$\text{cs}_k^{(r)} := \left<a,    -w_{k}^{(r)} F^\text{cl} \Delta u_{0:T-1}^{r} \right>$.
An example is given in Section~\ref{sec:experiment}.

\subsection{Extension to Constrained OCPs}\label{sec:consistency_extension_constraints}
Our discussion so far has centered around the unconstrained OCP \eqref{eq:OCP}. However, real-world systems may be subject to stage-wise constraints of the form $h_k^{(i)}(x_k,u_k) \leq 0$, where $h_k^{(i)}: \R^{\nstates} \times \R^{\ninputs} \to \R$, $i \in \{0,1,\ldots,I\}$, and $I$ is the number of constraints. In such a setup, the question of interest is: Given a desired correction $a$, would relaxing the $i$th constraint at stage $k$ improve the solution in the direction defined by the correction?

To reason about the effect of constraints, we propose an analysis framework similar to what is described above, but with an additional pre-processing step. Specifically, we begin by verifying whether the $i$th constraint is active, i.e., if $h_k^{(i)}(x_k,u_k) = 0$. If so, then one may re-express $h_k^{(i)}(x_k,u_k)$ in terms of control inputs only $\tilde h_k^{(i)}(u_{0:N-1})$ and compute the gradient $\nabla \tilde{h}_k^{(i)}(u_{0:N-1})$. If $\langle a, -F \nabla \tilde{h}_k^{(i)}(u_{0:N-1})   \rangle < 0$, then the $i$th constraint at stage $k$ is (locally) preventing the solution from improving in the desired correction direction. 

We point out that in general non-linear optimization, it is possible to construct examples where changing inactive constraints changes the solution. In our experience, however, given that many numerical solvers operate locally, we have observed that verifying activeness of a constraint works well.

\section{Applications of consistency scores}\label{sec:applications}

Next, we describe two applications of the consistency score from Definition~\ref{def:consistency_score}.
First, we show how one can use the concept of consistency to understand the solution of an OCP using directional corrections.
Then, we describe a method to automatically adjust parameters of the OCP to increase consistency with a set of directional corrections.

\subsection{OCP Consistency Analysis}\label{sec:consistency_analysis_application}

In practice, one can use consistency scores to understand which cost component(s) have the largest influence on an undesired behavior by providing a directional correction $a$ and then looking at the cost components with the lowest consistency scores.
We propose the following workflow: Suppose we are given an optimal, but undesirable, plan $\zeta = (x, u)$. Upon identifying the stage and direction in which we desire the plan to move, we construct the directional correction $a$ and compute the consistency scores as in Definition~\ref{def:consistency_score}.
A practitioner can then use the consistency scores to make useful inferences.
For example, $\sum_{k=0}^N \text{cs}_k^{(r)}$ gives the cumulative consistency of cost component $r$ over all stages with the provided correction.
Various aggregations allow for specific analysis, such as sorting the cost components by their consistency with a directional correction. 
Examples are provided in Section~\ref{sec:experiment}.

\subsection{OCP Parameter Optimization Algorithm}\label{sec:parameter_optimization}
Building on the notion of consistency, we now present a framework for optimally choosing weights $w_k^{(r)}$ of the cost components \eqref{eq:cost_components} such that the resulting cost function \eqref{eq:cost} is ``maximally" consistent with a set of plans (or closed-loop trajectories) and directional corrections.
This avoids tedious manual tuning of the weights, instead relying on a set of plans and directional corrections, $\{(\zeta_j, a_j)\}_{j\in J}$, from human annotations or generated automatically from requirements. 
Consider now the following {weight}-learning problem.
\begin{align}\label{eq:parameter_optimization}
\min_w \quad & \sum_{j\in J} \left[ \sum_{k,r} \max\left(m+\langle a_j, w_{k}^{(r)} F(\zeta_j)\nabla \tilde\ell_{k}^{(r)}(\zeta_j)   \rangle, 0 \right) \right] \nonumber \\
    \text{s.t.} \quad   &w \geq 0, \\
                        & \mathbf{1}^\top w = N+1, \nonumber \\
                        &w_{k+1}^{(r)} \leq w_k^{(r)}, \quad k=0,1,\ldots,N-1, \nonumber 
\end{align}
Similar to Section~\ref{sec:consistency}, let $F(\zeta_j)$ be the sensitivity matrix obtained by linearizing the dynamics about the $j$th trajectory and $\nabla \tilde\ell_{k}^{(r)}(\zeta_j)$ the gradient of the $r$th cost component at stage $k$ evaluated at $\zeta_j$.

Intuitively, for a single sample $(\zeta_j, a_j)$, the term inside the $\max(\cdot, 0)$ represents the \textit{inconsistency} score of the $r$th component at the $k$th stage with respect to the correction.
{The parameter $m > 0$ is the \emph{consistency margin}, and encourages strictly consistent plans satisfying
$\langle a_j, -w_{k}^{(r)} F(\zeta_j) \nabla \tilde\ell_{k}^{(r)}(\zeta_j) \rangle \geq m$.
As directional corrections do not indicate a desired magnitude of change, a small value of $m$ has been empirically observed to be sufficient.
Similarly, we minimize $\max(\cdot, 0)$ rather than the sum of inconsistency scores, which would aim to also increase the magnitude of the consistency score.}

The constraint $\mathbf{1}^\top w = N+1$ ensures that the weight vector remains bounded, since any positive multiple of $w$ results in the same optimizer.
Finally, we use the constraint $w_{k+1}^{(r)} \leq w_k^{(r)}$, similar to a discount factor in reinforcement learning, to avoid placing high importance on long-term actions over short-term ones.
Problem~\eqref{eq:parameter_optimization} is LP-representable and therefore can be solved efficiently even if the number of training data is large.
Different from other work \cite{jin2022learning} that uses corrections to optimize the weights, we penalize inconsistencies but do not prohibit them, since the presence of dynamics may not make every correction achievable.


Given an initial guess for the weights, $w^0$, the parameter-learning process alternates between data collection and parameter updates until convergence.
Intuitively, we build a set of trajectories for which we desire improvement and perform a batch optimization at each iteration to find $w$ that best satisfies the corrections on \emph{all} the trajectories.
A concise description is presented in Algorithm~\ref{alg:parameter_optimization}.


 \begin{algorithm}[tb]
 \caption{OCP parameter optimization algorithm}
 \label{alg:parameter_optimization}
 \begin{algorithmic}[1]
 \renewcommand{\algorithmicrequire}{\textbf{Input:}}
 \renewcommand{\algorithmicensure}{\textbf{Output:}}
 \REQUIRE OCP problem definition, max iterations $K$
 \ENSURE  Parameters $w$
 \\ 
 \textit{Initialize}: $k=0$, $w \in R^{N + 1}$ uniformly, dataset $\mathcal{D} = \{\}$, 
  \WHILE {$| \mathcal{D}^k | > 0$ and $k < K$}
  \STATE Compute trajectory $\zeta^k$ from $w^k$ by solving the OCP.
  \STATE Derive corrections $\mathcal{D}^k = \{a^k_i, \ldots\}$ from $\zeta^k$ and requirements.
  \STATE Append $\mathcal{D}^k$ to $\mathcal{D}$.
  \STATE Compute $w_{k+1}$ by solving \eqref{eq:parameter_optimization}.
  \STATE $k = k + 1$
  \ENDWHILE
 \RETURN $w^k$
 \end{algorithmic}
 \end{algorithm}
We demonstrate this process in Section~\ref{sec:experiment}.
Note that each training iteration may easily be extended to compute and generate corrections for many trajectories, allowing one to optimize the parameters over a diverse set of OCPs.

\subsubsection*{Relationship to Existing OCP Parameter Optimization Methods}
To close this section, we compare our method with the large body of work on cost-function learning for optimal control.
Our work addresses short-comings of existing methods when applied in practice to robotic systems.

First, a common approach assumes one is provided with expert demonstrations \cite{levine2012continuous, englert2017inverse}.
These methods aim to find a cost function whose optimal policy reconstructs the demonstrations.
In many real-world cases, however, acquiring expert demonstrations is costly (e.g. autonomous vehicles) or unnatural (e.g. high-DOF systems).

Alternatively, preference-based methods aim to find a cost function whose optimal policy is consistent with relative preferences indicated by an expert \cite{sadigh2017active}.
Unfortunately, in many robotics settings, generating two similar trajectories is not realistic.
For example, in a driving domain, one cannot replay real-world interactions starting from the same state.


Finally, some methods rely on corrections of generated trajectories similar to that analyzed in this work \cite{jin2022learning, bajcsy2017learning}. 
This choice reduces the annotation workload for a human-in-the-loop. 
However, the method of \cite{bajcsy2017learning} requires real-time physical interaction between the human and the robot, which can be impractical for robotic systems.
As noted, our work builds most directly on \cite{jin2022learning} by allowing directional corrections to be applied to states, handling closed-loop trajectories, and re-framing the learning problem as finding the most-consistent parameters with a set of directional corrections.

\section{Experiments}\label{sec:experiment}


We illustrate the proposed analysis method on three numerical examples. Section~\ref{sec:experiments_open_loop} demonstrates the use of consistency scores for analyzing open-loop plans; Section~\ref{sec:examples_closed_loop} studies a closed-loop example; and Section~\ref{sec:examples_parameter_optimization} showcases a parameter optimization problem. In all three examples, we consider a nonlinear system with non-convex cost functions.
{While this paper presents numerical examples, we have validated the framework on large quantities of real-world data.}

\subsection{Experimental setup}


We consider a motion planning problem for a robot modeled as a unicycle \cite{lavalle2006planning}, with $x = \begin{pmatrix} X, Y, \theta, v, \omega, a, \alpha\end{pmatrix}^{\top}$
and $u = \begin{pmatrix}j, \eta\end{pmatrix}^{\top}$. The (positional) states $p_r=(X,Y)^\top$ correspond to the center of the robot, while $\theta$ is the yaw angle with respect to the $X$-axis. $v$ and $\omega$ are the linear and angular velocities, respectively, $a$ and $\alpha$ are the linear and angular accelerations, and $j$ and $\eta$ are the linear and angular jerks. 
The continuous-time dynamics are given by $\dot{x} = \begin{bmatrix} v \cos(\theta), v \sin(\theta), \omega, a, \alpha, j, \eta\end{bmatrix}^{\top} \coloneqq \tilde f(x,u)$. 
The continuous-time dynamics can be brought into the form of \eqref{eq:dynamics} using a (forward) Euler discretization, such that $x_{k+1} = x_k + \tilde f(x_k,u_k)\,dt$, where $dt = 0.1$ seconds is the sampling time. 
The planning horizon is $N=50$. We solve the resulting OCPs using IPOPT \cite{wachter2006implementation}.

The objective function used in the following examples consists of cost components described in Table \ref{tab:cost_component_weights}.
The costs on tangential jerk, angular jerk, and lateral acceleration components are quadratic penalties of the robot's state variables. 
The next four components -- reference speed, reference path, obstacle, and boundary -- represent the structure of the environment.
Reference speed component is penalized as the squared error from a desired speed.
Reference path component is penalized as the squared Euclidean distance between the robot's position and the closest projected point on the path.
The boundary component encourages the robot to stay within a lateral distance of the reference path. 
A quadratic penalty is applied if the Euclidean distance to the closest point on the path to the robot's position $p_r=(X,Y)^\top$, $P_{\text{ref}}(p_r)$, exceeds a desired distance to the path, $d_w$.
The obstacle component similarly penalizes penetration into a sphere of radius $o_\text{buffer}$ around an obstacle's position $p_o$.
The final two components represent interactions with other agents.
We specifically consider the existence of a lead agent ahead of our robot, moving along our reference path, for which we wish to maintain a certain following distance.
The headway component penalizes violation of a desired time headway, $t_h$, to the lead agent. 
Given a position $p$, the function $D(p)$ computes its distance  along the reference path. 
Accordingly, the robot incurs a penalty if the distance it travels in time $t_h$ at velocity $v$ is larger than the distance the robot is behind the lead agent, along the reference path.
Finally, the relative speed component encourages the robot to follow the lead agent at a similar speed.



\subsection{Analyzing Open-Loop Plans}\label{sec:experiments_open_loop}
We demonstrate the consistency framework described in Section~\ref{sec:consistency_analysis_application} on a path-following problem, where the robot is asked to track a given reference path at a desired speed while avoiding static obstacles; see Fig.~\ref{fig:core_capability}. 
For this section, we use the nominal weights for each of the cost components noted in Table \ref{tab:cost_component_weights}.
The headway and relative speed costs noted in Table~\ref{tab:cost_component_weights} are not active as there is no lead agent in the scene.

\begin{table}[!tb]
\centering
\medskip
\begin{threeparttable}
\caption{Cost component definitions}
\label{tab:cost_component_weights}
\begin{tabular}{lrr} \toprule
Cost Component          & $\ell^{(r)}$      & Weight \\ \midrule
Tangential jerk         & $0.5j^2$     & 0.1 \\
Angular jerk            & $0.5\eta^2$     & 0.1 \\
Lateral acceleration    & $0.5 (v \omega)^2$        & 1.0 \\
Reference speed         & $0.5 (v - v_{\text{ref}})^2$        & 10.0 \\
Reference path          & $0.5 \lVert P_{\text{ref}}(p_r) - p_r \rVert_2^2$        & 100.0 \\
Obstacle                & $\sum_o 0.5 (o_{\text{buffer}} - \lVert p_o - p_r \rVert_2))_+^2 $        & 1000.0 \\
Boundary                & $0.5 (\lVert P_{\text{ref}}(p_r) - p_r \rVert_2 - d_{w})_{+}^2$        & 1000.0 \\
Headway                 & $0.5 ((t_{h} v) - (D(p_o) - D(p_r)))_{+}^2$        & 1000.0 \\
Relative speed         & $0.5 (v - v_{o})^2$        & 1.0 \\ \bottomrule
\end{tabular}
\begin{tablenotes}
\item The $(x)_+^2$ notation indicates a one-sided quadratic where the penalty is applied only if $x > 0$.
\end{tablenotes}
\vspace{-2\baselineskip}
\end{threeparttable}
\end{table}

As an illustration, we construct an example with a curved path as shown in Fig.~\ref{fig:core_capability}.
This example simulates an issue {that may stem from an improperly sensed obstacle, or a sudden change to the robot's environment.}
The presence of this object, shown as a red X in the top subplot near coordinates $(40, 2.5)$, results in the orange trajectory (top subplot). 
We see from the center subplot that the robot slows down significantly from the desired speed of 10m/s.
While in this example it is simple to infer that the obstacle and its associated cost are problematic, the complexity of a real-world system can make that determination challenging.
For example, obstacles could be small and difficult to see in a visualization, or obscured by other contents of the scene, but are nonetheless detected by sensors and sent to a motion planner.

We now show how our method can be leveraged to identify the cause of the slowdown. To this end, we apply a directional correction at stage $10$ to indicate a desired increase in speed, i.e., $a_{x_{10}} = (0, 0, 0, 1, 0, 0, 0)$, with all other elements of $a$ set to zero. This is shown as the red arrow in the center subplot.
Fig.~\ref{fig:core_capability} (bottom right) shows the consistency score of the individual cost components, accumulated over the horizon, i.e., $\text{cs}^{(r)}=\sum_k \text{cs}_k^{(r)}$.
We see that the \texttt{OBSTACLE} component is \emph{least consistent} with the desired correction of increasing speed, while the \texttt{REFERENCE\_SPEED} component is \emph{most consistent} with the desired correction.
Intuitively, the obstacle cost pushes most strongly against a higher speed, while the reference speed cost most strongly supports it.
To verify the result, we remove the obstacle cost by zeroing the associated weights. 
The resulting, blue trajectory is shown in Fig.~\ref{fig:core_capability} (center). 
Specifically, we observe that by disabling the weight on the inconsistent obstacle cost, the robot's speed profile indeed moves in the direction indicated by the correction.


\subsubsection*{Comparison to Cost-Based Analysis}
A common approach to understanding a plan's behavior is to analyze the individual cost terms $\ell_k^{(r)}$ themselves, and attribute deficiency to cost components that take large values $\ell_k^{(r)}(x^\star, u^\star)$. We demonstrate next that such an approach may lead to the wrong conclusion:
Fig.~\ref{fig:core_capability} (bottom left) shows the objective function $J(x_{0:N}^\star, u_{0:N-1}^\star)$ broken down in the individual cost components at the stage of the correction, i.e., $w_{10}^{(r)} \ell_{10}^{(r)}(x_{10}^\star,u_{10}^\star)$. 
We see from the plot that \texttt{REFERENCE\_SPEED} and \texttt{REFERENCE\_PATH} are the components with highest value, due to the robot slowing down and deviating from the reference path to avoid the obstacle. 
Furthermore, the (problematic) obstacle cost component is zero at the correction stage as the robot is outside the obstacle buffer at $k=10$.
Note how, in contrast, our consistency-based analysis inherently accounts for cross-stage coupling through the sensitivity matrix \eqref{eq:sensitivity}, and correctly links the slowdown and correction at stage $10$ to the obstacle almost $40$ stages later. 
In addition, even if one observes the cumulative values of the cost components over all stages, $\sum_k w_k^{(r)} \ell_k^{(r)}(x_k^\star,u_k^\star)$, the \texttt{REFERENCE\_SPEED} component remains an order of magnitude larger than the \texttt{OBSTACLE} component.
Moreover, the cumulative value cannot communicate the impact of a component at a certain stage or whether that impact is consistent with the desired outcome.

To summarize, an analysis based on cost value alone can lead to inaccurate conclusions. 
In contrast, the proposed consistency-based analysis produces more accurate signals by leveraging first-order information and accounting for cross-stage effects using the system dynamics.
This is not surprising since the optimality condition is not defined with respect to the value of a cost function $J(\cdot)$, but rather its gradient $\nabla J(\cdot) = 0$ \cite[Section~5.5.3]{boyd2004convex}.

\begin{figure}[tb]
\centering
\includegraphics[width=\columnwidth]{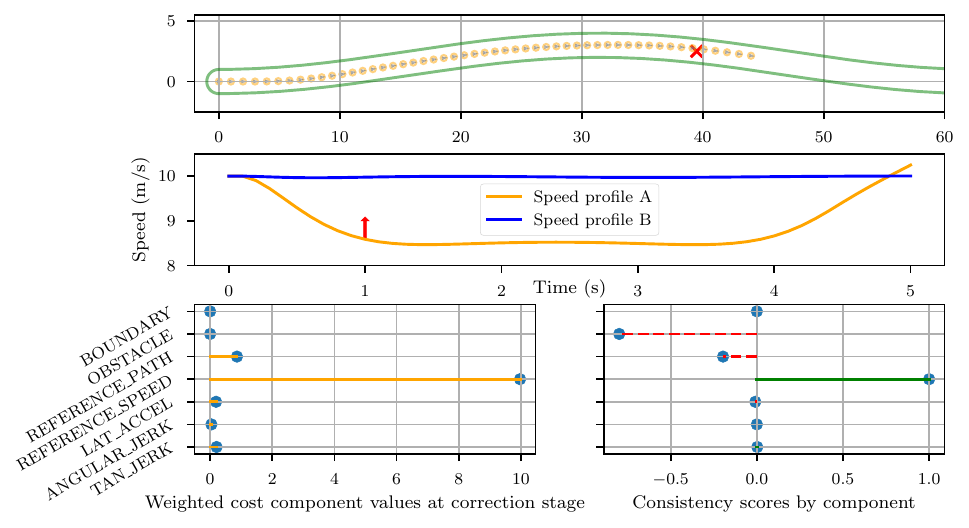}
\caption{Example of open-loop OCP corrections analysis. The top plot shows the open-loop evolution of the robot {within a corridor (green boundary)} with an improperly sensed obstacle near ({40}, 2.5). The center plot shows the speed profile of the robot when affected by the obstacle (orange) and the desired change to that speed profile. The blue profile is the result when the obstacle is removed. The bottom plot shows the computed costs and consistency scores of each component. The consistency analysis successfully detects the obstacle component as inconsistent (red dashed line) even though the magnitude of the obstacle cost is small.}
\vspace{-\baselineskip}
\label{fig:core_capability}
\end{figure}

\subsection{Analyzing Closed-Loop Trajectories}\label{sec:examples_closed_loop}

Next we show an application of the consistency framework to trajectories generated by an MPC scheme.
Assuming a duration of $T=30$ cycles, $T$ plans are computed which result in $T$ controls being executed and a trajectory with $T+1$ states. 
We augment our example above with a lead agent in the environment, such as another robot or a human, for which the controlled robot generates predictions.
The static obstacle from the previous example is removed.
At each planning time, a prediction of the other agent is generated and used in the OCP as a time-varying obstacle, to ensure an acceptable interaction.
In this example, we enable both the cost on headway, which penalizes the robot for being within a one-second time-headway of the lead agent, and the cost on relative speed, which encourages the robot to stay close to the lead agent's speed.

To demonstrate the practical use of this analysis, we simulate a noisy prediction for the lead agent.
We initialize the simulation with the robot at its desired speed and desired headway.
{In reality, the lead agent proceeds with a constant velocity equal to our robot's desired speed.}
That means there is no reason for the robot to respond to the agent to satisfy the desired headway if the predictions were accurate.
However, for a few planning frames $t = \{10, \ldots, 19\}$ we generate a prediction for the lead agent that incorrectly decelerates with $-1.0\, m/s^2$.
With this setup, as one might expect, we observe a modest deceleration in the closed-loop trajectory around the frames with inaccurate predictions.
The scenario is depicted in Fig.~\ref{fig:cl_slow_down}.

\begin{figure}[tb]
\centering
\includegraphics[width=\columnwidth]{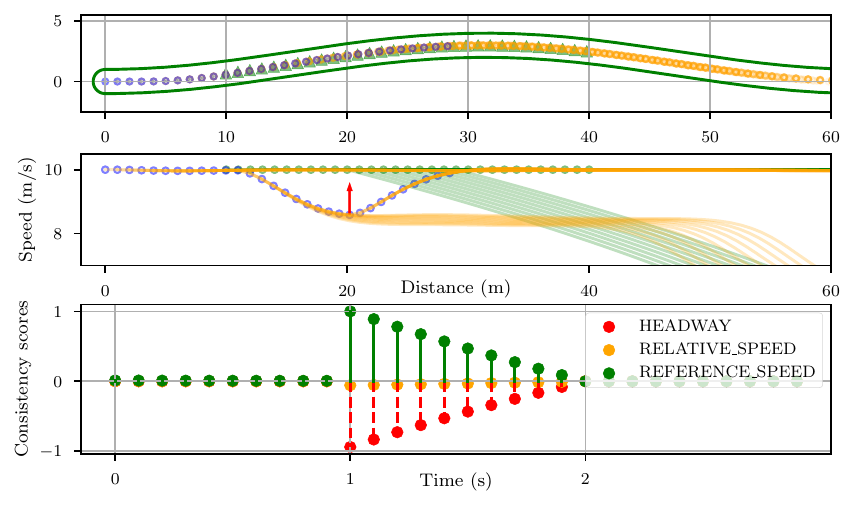}
\caption{Example of closed-loop OCP corrections analysis. The top plot shows the closed-loop evolution of the robot (blue circles) and lead agent (green triangles) {within a reference corridor (green boundary).} The orange trajectories indicate the computed -- but not executed -- plans of the robot.
The center plot shows the speed profile of the robot and its response to an incorrect decelerating prediction of the lead agent. 
The bottom plot shows the consistency scores of three components of particular interest. The consistency analysis detects the headway and relative speed components as inconsistent at all times for which there is an incorrect prediction.}
\label{fig:cl_slow_down}
\vspace{-\baselineskip}
\end{figure}

Suppose we are interested in understanding the cause of the slowdown, given the closed-loop executed trajectories by both our robot and the lead agent.
In this case, the cause is not obvious as the lead agent does not slow down, but only its predictions.
This setup simulates situations in practice where, when aiming to understand the cause of undesired robot behavior, one is not usually told a priori that a prediction was incorrect.
Using our framework, we apply a correction at $t=21$ to increase speed, since it is the time with maximum error relative to the desired speed of $10\, m/s$.

The bottom subplot in Fig.~\ref{fig:cl_slow_down} shows the consistency scores $cs_t^{(r)}$ of three cost components of interest, for each the planning cycle.
Each time point in this plot shows the sum of the consistency scores of a cost component over the plan at that cycle.
In this example, the analysis indicates that headway is the predominant component causing the braking behavior.
Our analysis framework detects the headway component as problematic at the first planning time for which there are incorrect predictions ($t=10$), and it remains detected through the final time of incorrect prediction ($t=19$).
As before, the respective costs are not indicative of the relative impact. Indeed,
the headway cost has a very small magnitude for all planning cycles that have a decelerating prediction ($< 0.001$), while the values of the relative speed and reference speed components remains one and two orders of magnitude larger.
In other words, large cost component values do not necessarily indicate that the corresponding components caused the problem.

{
This example illustrates a scenario where the robot's undesired behavior is caused by poor predictions. In our experience, we have found the proposed method particularly useful in cases that are not easily understood by visual inspection.
We point out that our framework would have also led to similar conclusions had the lead agent indeed decelerated.}

\subsection{OCP Parameter Optimization}\label{sec:examples_parameter_optimization}
This section demonstrates the parameter optimization capability described in Section~\ref{sec:parameter_optimization} in both open- and closed-loop.

\subsubsection{Open-Loop Parameter Optimization}
First, we show results on an open-loop problem that aims to find the parameters of the OCP that result in the optimal plan having desired behavior.
For this experiment, we use the same scenario as in Section \ref{sec:experiments_open_loop}, in which the robot must follow a reference path at a desired speed while staying away from the boundaries. 
In this case, we initialize the weights $w \in \R^{(N + 1)R}$ uniformly with planning horizon of $N = 50$, and $R = 6$ cost components.
We assume that no other agents or obstacles exist.

{We apply the procedure described in Algorithm~\ref{alg:parameter_optimization}} to find the optimal weights and automatically construct corrections on each trajectory based on requirements to stay within $0.25m/s$ of the reference speed and $0.25m$ of the reference path.
{The initial speed is $8m/s$} and we require the robot to reach $10 \pm 0.25 m/s$ by the final $35$ stages of the plan. The reference path corrections are applied over the entire plan.
We generate a single trajectory from the OCP with the latest weights on each iteration.
The parameter update step \eqref{eq:parameter_optimization} is modeled in CVXPY \cite{diamond2016cvxpy} and solved by the Clarabel solver \cite{Clarabel_2024} with a margin $m=0.001$.
In this example, we set the headway and obstacle component weights to be zero as these terms have no effect.

{The while-loop of Algorithm~\ref{alg:parameter_optimization}} is run for 22 iterations after which the solution of OCP is a plan with no corrections. Fig.~\ref{fig:ol_training_history} shows the learned weights and the errors over iterations. 
As expected, the weights on the boundary, reference path and reference speed are increased (left subplot) to keep the robot near the path and reach the desired speed.

\begin{figure}[tb]
\centering
\includegraphics[width=\columnwidth]{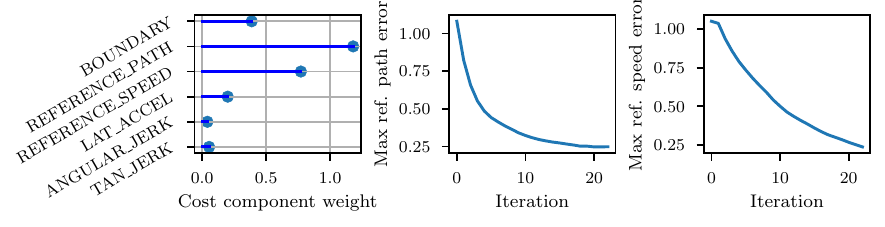}
\caption{Illustration of the open-loop training process. 
The left plot shows optimal weights for each cost component aggregated over all stages.
The right two plots show the improvement in reference path and reference speed error through the training process.}
\vspace{-\baselineskip}
\label{fig:ol_training_history}
\end{figure}

\subsubsection{Closed-Loop Parameter Optimization}
Next, we apply parameter optimization to closed-loop tuning.
In this case, the goal is to find parameters such that the closed-loop trajectory is optimal, rather than a single open-loop plan.
To showcase this difference, we construct a scenario similar to the one in Fig.~\ref{fig:cl_slow_down}; however, we delay the lead agent's deceleration to the final 40 stages of the prediction, leaving one second of correct prediction. This simulates the case in practice where short-horizon predictions are often more accurate than long-term predictions.
The goal in this scenario is to maintain a constant headway behind the lead agent in closed-loop, despite a faulty long-term prediction.
Specifically, a correction is applied to the closed-loop trajectory at any time where the true headway (not the predicted one) differs by more than $\pm 0.1m$.
Intuitively, this example shows how one may optimize parameters to achieve the desired closed-loop behavior in the presence of noisy data inputs (in this case, agent predictions) at each open-loop planning cycle {and reduce sensitivity to environmental changes toward the end of the planning horizon.}

We initialize the robot with a speed of $10$\,m/s and the desired headway of 1 second behind the lead agent. 
Despite the incorrect decelerating predictions, the lead agent proceeds at a constant speed of $10m/s$.
Unlike the previous examples, however, we reduce the reference speed to $8m/s$.
This change to reference speed makes the task more challenging as one needs to find a trade-off between the headway, relative speed, and reference speed terms that result in a nearly constant headway in closed-loop even though none of the terms alone will produce the desired outcome.

Here, we expect to find parameters that discount the effect of later predictions in order to maintain the desired closed-loop behavior.
Figure~\ref{fig:cl_training_history} shows the result over 30 training iterations.
The {left and center} plots show that the learning process converged to a solution that is very close to a constant speed as desired.
The {right} plot shows the weights resulting from the optimization.
Note that both of the prediction-dependent components have a strong influence at the beginning of the plan.
In particular, the relative speed component must use the correct part of the prediction (1s) in order for the robot to stay at the same speed as the lead agent.
The optimal weights on both headway and relative speed quickly become small after stage 10, where the incorrect prediction begins.
At convergence, the closed-loop headway distance error is less than 0.15 meters.

\begin{figure}[tb]
\centering
\includegraphics[width=\columnwidth]{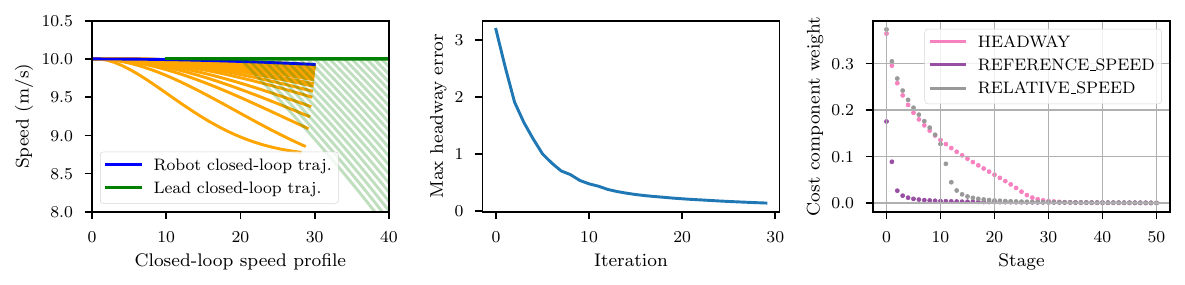}
\caption{Illustration of the closed-loop training process. The left plot shows the optimal trajectory at each training iteration (orange), along with the predicted trajectory of the lead agent (light green), and the final closed-loop trajectories for the robot (blue) and lead agent (dark green). The center plot shows the error in desired headway of the closed-loop trajectory. This error trends to zero as the robot finds a parameter configuration that avoids deceleration. The right plot shows the optimized parameter values by stage for the three cost components of most relevance.}
\vspace{-1\baselineskip}
\label{fig:cl_training_history}
\end{figure}

\section{Conclusion}
In this paper, we presented a method to understand how cost components impact the solution of optimal control problems (OCPs).
Our analysis relies on a directional correction indicating the desired change in the optimal plan produced by the OCP.
First, we provided a theoretical justification for the use of corrections in OCP analysis.
Next, we demonstrated the utility of our method to understand the cause of undesired behavior in both open- and closed-loop scenarios.
In addition, for a typical linear parametrization of the cost function, we described a method to optimize the cost component weights to make the OCP consistent with a set of directional corrections.


\section*{Acknowledgment}

The authors would like to thank Chung Yen Lin and Royce Cheng-Yue for their continued support, and Joao Moreira de Sousa Pinto for insightful discussions.



\bibliographystyle{IEEEtran}
\bibliography{IEEEabrv,references}

\begin{thebibliography}{10}
\providecommand{\url}[1]{#1}
\csname url@samestyle\endcsname
\providecommand{\newblock}{\relax}
\providecommand{\bibinfo}[2]{#2}
\providecommand{\BIBentrySTDinterwordspacing}{\spaceskip=0pt\relax}
\providecommand{\BIBentryALTinterwordstretchfactor}{4}
\providecommand{\BIBentryALTinterwordspacing}{\spaceskip=\fontdimen2\font plus
\BIBentryALTinterwordstretchfactor\fontdimen3\font minus
  \fontdimen4\font\relax}
\providecommand{\BIBforeignlanguage}[2]{{%
\expandafter\ifx\csname l@#1\endcsname\relax
\typeout{** WARNING: IEEEtran.bst: No hyphenation pattern has been}%
\typeout{** loaded for the language `#1'. Using the pattern for}%
\typeout{** the default language instead.}%
\else
\language=\csname l@#1\endcsname
\fi
#2}}
\providecommand{\BIBdecl}{\relax}
\BIBdecl

\bibitem{lavalle2006planning}
S.~M. LaValle, \emph{Planning algorithms}.\hskip 1em plus 0.5em minus
  0.4em\relax Cambridge univ. press, 2006.

\bibitem{murray2009optimization}
R.~M. Murray \emph{et~al.}, ``Optimization-based control,'' \emph{California
  Institute of Technology, CA}, pp. 111--128, 2009.

\bibitem{padenFrazzoli2016survey}
B.~Paden, M.~{\v{C}}{\'a}p, S.~Z. Yong, D.~Yershov, and E.~Frazzoli, ``A survey
  of motion planning and control techniques for self-driving urban vehicles,''
  \emph{IEEE Transactions on intelligent vehicles}, vol.~1, no.~1, pp. 33--55,
  2016.

\bibitem{Zhang2021_OBCA}
X.~Zhang, A.~Liniger, and F.~Borrelli, ``Optimization-based collision
  avoidance,'' \emph{IEEE Transactions on Control Systems Technology}, vol.~29,
  no.~3, pp. 972--983, 2021.

\bibitem{oldewurtel2012buildingClimateControl}
F.~Oldewurtel \emph{et~al.}, ``Use of model predictive control and weather
  forecasts for energy efficient building climate control,'' \emph{Energy and
  buildings}, 2012.

\bibitem{ZhangECC2013_BuildingControl}
X.~Zhang, G.~Schildbach, D.~Sturzenegger, and M.~Morari, ``Scenario-based mpc
  for energy-efficient building climate control under weather and occupancy
  uncertainty,'' in \emph{2013 European Control Conference (ECC)}, 2013, pp.
  1029--1034.

\bibitem{maddalena2020data}
E.~T. Maddalena, Y.~Lian, and C.~N. Jones, ``Data-driven methods for building
  control—a review and promising future directions,'' \emph{Control
  Engineering Practice}, vol.~95, p. 104211, 2020.

\bibitem{moura2010stochastic}
S.~J. Moura, H.~K. Fathy, D.~S. Callaway, and J.~L. Stein, ``A stochastic
  optimal control approach for power management in plug-in hybrid electric
  vehicles,'' \emph{IEEE Transactions on control systems technology}, vol.~19,
  no.~3, pp. 545--555, 2010.

\bibitem{chaturvedi2010algorithms}
N.~A. Chaturvedi, R.~Klein, J.~Christensen, J.~Ahmed, and A.~Kojic,
  ``Algorithms for advanced battery-management systems,'' \emph{IEEE Control
  systems magazine}, vol.~30, no.~3, pp. 49--68, 2010.

\bibitem{guanetti2017optimal}
J.~Guanetti, S.~Formentin, M.~Corno, and S.~M. Savaresi, ``Optimal energy
  management in series hybrid electric bicycles,'' \emph{Automatica}, vol.~81,
  pp. 96--106, 2017.

\bibitem{Borrelli_Bemporad_Morari_MPC_book_2017}
F.~Borrelli, A.~Bemporad, and M.~Morari, \emph{Predictive Control for Linear
  and Hybrid Systems}.\hskip 1em plus 0.5em minus 0.4em\relax Cambridge
  University Press, 2017.

\bibitem{Mayne_nonlinearMPC2000}
D.~Mayne, ``Nonlinear model predictive control: Challenges and opportunities,''
  in \emph{Nonlinear Model Predictive Control}, 2000.

\bibitem{jin2022learning}
W.~Jin, T.~D. Murphey, Z.~Lu, and S.~Mou, ``Learning from human directional
  corrections,'' \emph{IEEE Transactions on Robotics}, 2022.

\bibitem{boyd2004convex}
S.~P. Boyd and L.~Vandenberghe, \emph{Convex optimization}.\hskip 1em plus
  0.5em minus 0.4em\relax Cambridge university press, 2004.

\bibitem{gal2012advances}
T.~Gal and H.~J. Greenberg, \emph{Advances in sensitivity analysis and
  parametric programming}.\hskip 1em plus 0.5em minus 0.4em\relax Springer
  Science \& Business Media, 2012.

\bibitem{bemporad2002explicit}
A.~Bemporad, M.~Morari, V.~Dua, and E.~N. Pistikopoulos, ``The explicit linear
  quadratic regulator for constrained systems,'' \emph{Automatica}, vol.~38,
  no.~1, pp. 3--20, 2002.

\bibitem{borrelli2017predictive}
F.~Borrelli, A.~Bemporad, and M.~Morari, \emph{Predictive control for linear
  and hybrid systems}.\hskip 1em plus 0.5em minus 0.4em\relax Cambridge
  University Press, 2017.

\bibitem{levine2012continuous}
S.~Levine and V.~Koltun, ``Continuous inverse optimal control with locally
  optimal examples,'' \emph{arXiv preprint arXiv:1206.4617}, 2012.

\bibitem{englert2017inverse}
P.~Englert, N.~A. Vien, and M.~Toussaint, ``Inverse kkt: Learning cost
  functions of manipulation tasks from demonstrations,'' \emph{The
  International Journal of Robotics Research}, 2017.

\bibitem{sadigh2017active}
D.~Sadigh, A.~D. Dragan, S.~Sastry, and S.~A. Seshia, ``Active preference-based
  learning of reward functions,'' \emph{Robotics: Science and Systems XIII},
  2017.

\bibitem{bajcsy2017learning}
A.~Bajcsy, D.~P. Losey, M.~K. O’malley, and A.~D. Dragan, ``Learning robot
  objectives from physical human interaction,'' in \emph{Conference on robot
  learning}.\hskip 1em plus 0.5em minus 0.4em\relax PMLR, 2017, pp. 217--226.

\bibitem{wachter2006implementation}
A.~W{\"a}chter and L.~T. Biegler, ``On the implementation of an interior-point
  filter line-search algorithm for large-scale nonlinear programming,''
  \emph{Mathematical programming}, vol. 106, pp. 25--57, 2006.

\bibitem{diamond2016cvxpy}
S.~Diamond and S.~Boyd, ``{CVXPY}: {A} {P}ython-embedded modeling language for
  convex optimization,'' \emph{Journal of Machine Learning Research}, vol.~17,
  no.~83, pp. 1--5, 2016.

\bibitem{Clarabel_2024}
P.~J. Goulart and Y.~Chen, ``Clarabel: An interior-point solver for conic
  programs with quadratic objectives,'' 2024.

\end{thebibliography}

\end{document}